\documentclass[10pt,conference]{IEEEtran}
\usepackage{subfig}
\usepackage{bm}
\usepackage{mathrsfs}
\usepackage{textcomp}
\usepackage{epsfig}
\usepackage{graphicx}
\usepackage{epstopdf}
\usepackage{indentfirst}
\usepackage{amsmath}
\usepackage{amssymb}
\usepackage{array}
\usepackage{multirow}
\usepackage{graphicx}
\usepackage{color}

\usepackage{amsthm}
\newtheorem{theorem}{Theorem}
\newtheorem{lemma}[theorem]{Lemma}

\usepackage{times}
\usepackage[ruled,vlined]{algorithm2e}
\theoremstyle{definition}

\theoremstyle{remark}

\usepackage{algorithmic}
\usepackage{cite}

\def\BibTeX{{\rm B\kern-.05em{\sc i\kern-.025em b}\kern-.08em T\kern-.1667em\lower.7ex\hbox{E}\kern-.125emX}}
\usepackage{graphicx}

\begin{document}
\title{Optimizing NOMA Transmissions to Advance Federated Learning in Vehicular Networks}


\author{
    \IEEEauthorblockN{Ziru Chen\IEEEauthorrefmark{1}, Zhou Ni\IEEEauthorrefmark{2}, Peiyuan Guan\IEEEauthorrefmark{3}, Lu Wang\IEEEauthorrefmark{4}, Lin X. Cai\IEEEauthorrefmark{1}, Morteza Hashemi\IEEEauthorrefmark{2}, Zongzhi Li\IEEEauthorrefmark{4}}
\IEEEauthorblockA{\IEEEauthorrefmark{1}Department of Electrical and Computer Engineering, Illinois Institute of Technology, Chicago, USA,\\ \IEEEauthorrefmark{2}Department of Electrical Engineering and Computer Science, University of Kansas, Lawrence, USA, \\
\IEEEauthorrefmark{3} Department of Informatics, University of Oslo, Oslo, Norway,\\
\IEEEauthorrefmark{4} Department of Civil Engineering, Illinois Institute of Technology, Chicago, USA}

\IEEEauthorblockA{zchen71@hawk.iit.edu, zhou.ni@ku.edu, peiyuang@ifi.uio.no, lwang121@hawk.iit.edu, }

\IEEEauthorblockA{lincai@iit.edu, mhashemi@ku.edu, lizz@iit.edu}
}

\maketitle

\begin{abstract}
Diverse critical data, such as location information 
and driving patterns, can be collected by IoT devices in vehicular networks to improve driving experiences and road safety. 
However, drivers are often reluctant to share their data due to privacy concerns. The Federated Vehicular Network (FVN) is a promising technology that tackles these concerns by transmitting model parameters instead of raw data, thereby protecting the privacy of drivers. 
Nevertheless, the performance of Federated Learning (FL) in a  vehicular network depends on the joining ratio, which is restricted by the limited available wireless resources. 
To address these challenges, this paper proposes to apply Non-Orthogonal Multiple Access (NOMA) to improve the joining ratio in a FVN. Specifically, a vehicle selection and transmission power control algorithm is developed to exploit the power domain differences in the received signal to ensure the maximum number of vehicles capable of joining the FVN. 
Our simulation results demonstrate that the proposed NOMA-based strategy increases the joining ratio and significantly enhances the performance of the FVN. 






\end{abstract}

\begin{IEEEkeywords}
Federated Vehicular Network, NOMA
\end{IEEEkeywords}

\section{Introduction}


Advancements in Internet of Things (IoT) technologies have made vehicles more complex and intelligent over the past decades. As a result, vehicles generate a wide variety of data from devices such as engines and radars. The collected data from vehicles serves as a foundation for developing innovative technologies that enhance the driving experience and improve road safety~\cite{ye2024multiplexed}. 
However, privacy concerns, as pointed out in~\cite{lu2018survey}, pose a significant challenge to data exchange and collection as drivers may be reluctant to share sensitive information such as location or driving behavior with others.


Federated learning (FL) is a promising technology to address privacy concerns. 
It involves model training across decentralized devices such as vehicles, which hold local data samples without exchanging them. By integrating FL into vehicular networks, federated vehicular networks (FVN) have emerged to enable vehicles to collaboratively learn and adapt to new environments while ensuring privacy~\cite{liu2024fedagl}. However, implementing FL in vehicular networks poses new challenges~\cite{lai2020security}, including selecting appropriate FL clients, ensuring robust communication between vehicles, and managing the learning process efficiently across the vehicular network.


Existing works \cite{nguyen2021federated, chen2021communication} on federated learning mainly study the impacts of learning parameters such as learning rate, communication rounds, local epochs, and so on, on the convergence speed, communication efficiency, and performance of the federated learning model. One under-explored essential metric in FL is the joining ratio, which represents the proportion of clients participating in each round of training. Generally, A higher joining ratio leads to a better performance of federated learning, as more clients contribute to the model parameters, allowing the global model to converge faster and improve its accuracy~\cite{liu2022joint,nguyen2020efficient,xu2021learning}. When the joining ratio reaches $1$, it indicates that all clients are actively involved in the training process for that round, maximizing the collective learning potential of the system. 
However, the mobility of participant vehicles and dynamic network typologies make it very challenging to coordinate FL tasks and maintain a stable set of participating vehicles over time. Thus, it is crucial to improve the joining ratio in a FVN to enhance the efficiency of distributed learning algorithms.




Non-Orthogonal Multiple Access (NOMA) is a cutting-edge multiple access technique for 5G and beyond. Unlike the traditional multiple access technologies where users are allocated orthogonal network resources for channel access, e.g., orthogonal time slots or orthogonal subcarriers, NOMA allow users to transmit concurrently using the same time-frequency resource. Thus, NOMA promotes spectrum efficiency and multiple access by allowing users to share the same resource non-orthogonally. The superimposed signals with different power levels are then sequentially decoded at the receiver side based on successive interference cancellation (SIC) technique~\cite{chen2021performance, ding2017survey, chen2021performance1}. 

In this work, we propose to apply advanced NOMA communications to augment the joining ratio and fostering collaborative and efficient federated learning in vehicular networks. To this end, we first propose an FVN architecture that utilizes NOMA for vehicle communications. Leveraging dynamic vehicle trajectories, the transmit power control and vehicle selections are jointly optimized to maximize the participation ratio of vehicles in federated learning. To the best of our knowledge, this is the first work that applies advanced NOMA communications to improve the performance of a FVN. The contributions of this paper are listed as follows.

\begin{enumerate}
\item We formulate a NOMA-enabled FVN framework which incorporates advanced NOMA transmissions to promote vehicles' participation in distributed learning.

\item We propose a joint vehicle selection and transmission power control algorithm to maximize the number of vehicles that can participate in the model training, while mitigating the mutual interference among concurrent transmissions of multiple vehicles.  


\item 
The proposed  NOMA-enabled FVN framework greatly promotes the joining ratio of FL, leading to accelerated convergence speed and enhanced learning stability for both i.i.d. and non-i.i.d.  data sets.  


\end{enumerate}



The remainder of this paper is organized as follows. The system model is presented in Section II. A joint vehicle selection and power allocation scheme is proposed in Section III, followed by the FL algorithm in Section IV.  Performance evaluation is presented in Section V. Finally, concluding remarks are provided in Section VI.

\label{sec:label}

\section{System Model}  \label{sec:sys}
\subsection{NOMA-enabled FVN framework}

We consider a wireless FVN with a set of $N$ vehicles, i.e., $ \{1, 2, ..., N\}$ and a single BS within a given area as shown in Fig.~\ref{fig::sys}, where BS is the central server, and vehicles are clients participate in a FL task, such as street view recognition, traffic control, etc. Denote the dataset of the $n$-th vehicle as $D_{n} = { (\textbf{x}_{i}^{(n)}, y_{i}^{(n)}) }_{i = 1}^{k_n}$, drawn from a distribution $\mathcal{D}_{n}$ over the space $\mathcal{X} \times \mathcal{Y}$. The total dataset size across all vehicles is denoted by $K = \sum_{n = 1}^{N} k_n$, where $k_n$ is the number of samples of vehicle $n$. Thus, the objective of FL can be written as:
\begin{equation}
\min_{\omega^{*}} \frac{1}{N} \sum_{n=1}^{N} \mathbb{E}_{\mathcal{D}_{n} \sim \textbf{x}_{i}^{(n)}, y_{i}^{(n)}} [f_{n} ( \omega_{n}; D_{n} )],
\end{equation}
where $\textit{f}_{n}: \mathcal{X}\times \mathcal{Y} \to \mathbb{R}_{+}$ denotes the loss function for the $n$-th vehicle, and $\omega$ is the model parameters during the training. The objective of the training process is to obtain the optimal global model parameters to achieve the best learning performance. To achieve this, the BS communicates with multiple vehicles to exchange model parameters. To enable efficient data exchange, NOMA is adopted for uplink transmissions of vehicles. Due to high interference among transmitting clients over wireless channels, it is optimal to select a subset of vehicles for uplink NOMA transmissions~\cite{jee2022performance}. As such, in the downlink, the BS broadcasts a global model to vehicles within its coverage area, the selected subset of vehicles $U_t$ to participate in the NOMA-based FL training, and the corresponding transmission powers of the participating vehicles. After several epochs of local training, the selected vehicles upload their updated local models, i.e., $\omega^{*}$ to the BS in the uplink NOMA transmissions.  


\subsection{Data Communication Model}
\begin{figure}[t]    \includegraphics[width=0.37\textwidth]{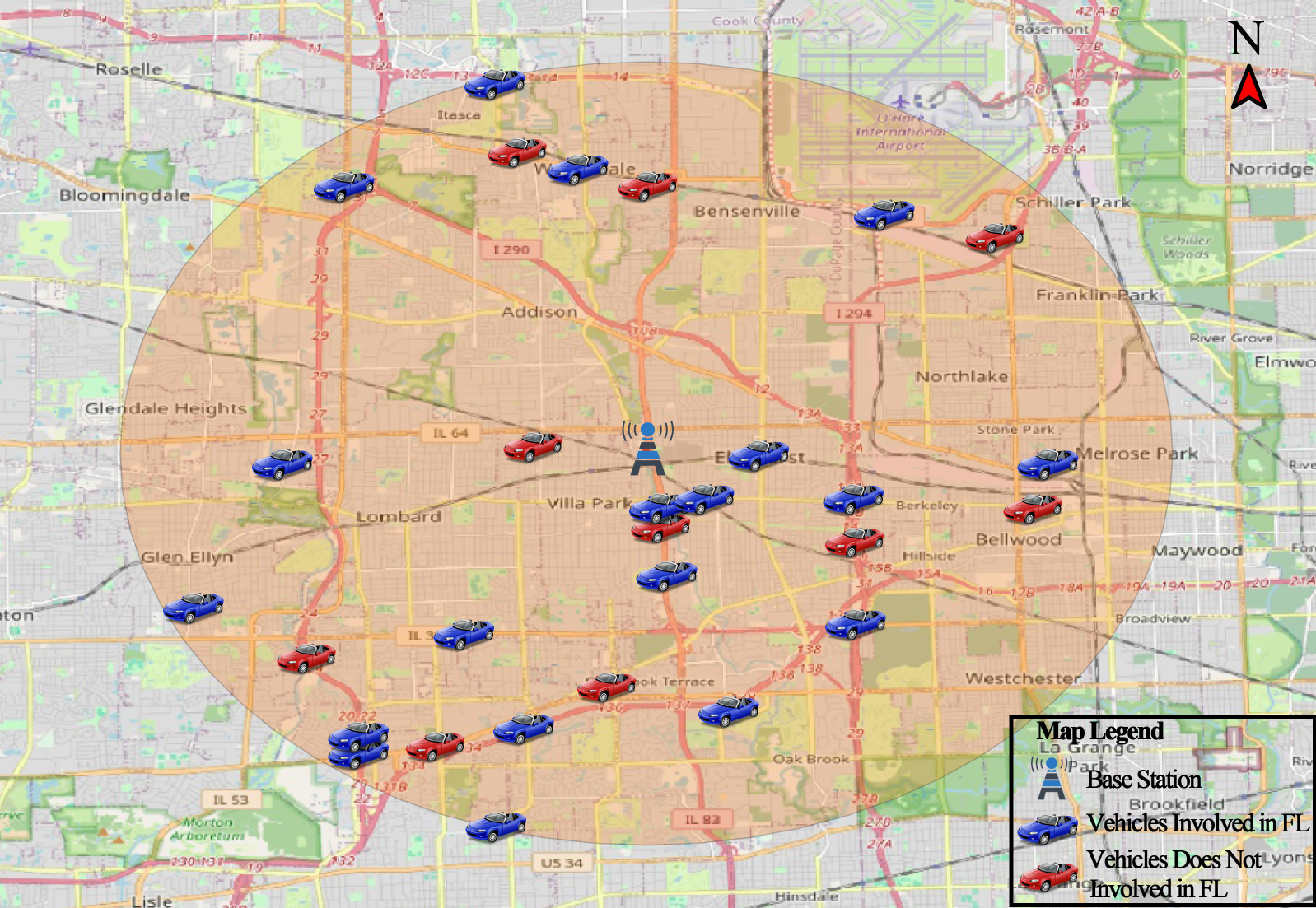}
    \centering
    \caption{System Model Illustration.}
    \label{fig::sys}
\end{figure}

Throughout the uplink FL aggregation phase, vehicles are actively trying to transmit their local model information to the BS utilizing the wireless communication channel. Consequently, the signal received at the BS for any given time slot, denoted by $I_t$ is thus given by: 
\begin{align}
   I_t=\sum_{n=1}^{M} h_{n,t}\sqrt{p^t_{n,t}}x_{n,t}+\sigma_t, 
\end{align}
 in which $M\leq N$ is the number of vehicles that transmit the signals at the same resource block, $h_{n,t}$ is the channel state information (CSI)  between the $n$-th vehicle to the BS in time slot $t$. $p^t_{n,t}$ and $x_{n,t}$ are the transmit power and transmission  signal of the vehicle $n$ in slot $t$ respectively, and $\sigma_t$ is the white Gaussian noise. 

NOMA is employed during uplink transmissions, which allows multiple vehicles to simultaneously transmit their updates using the same resource block, thereby improving the efficiency of resource usage and ensuring effective data communication among multiple vehicles. Denote the received signal strength of vehicle $n$ as
\begin{equation} \label{eq:qkt1}
p^r_{n,t} = p^t_{n,t}h_{n,t}^{2}, \ \forall n. 
\end{equation}
Note that the transmit power is upper bounded, where $p^t_{n,t}\leq p^t_{n,max}$.
The BS applies SIC to decode the received signals sequentially, prioritizing them in descending order of their signal strengths., i.e.,
$p^r_{1,t}\geq p^r_{2,t},\geq\dots\geq p^r_{M,t}$. 
The SINR of the $n$-th signal is thus given by~\cite{chen2022deep}: 
\begin{equation}\label{eq:SINR_before}
 \mathbf {SINR}_{n,t}=\frac{p^r_{n,t}}{\sum_{j=1}^{n-1} a_{j,t} p^r_{j,t}+\sum_{i=n+1}^{M} p^r_{i,t}+\sigma_t^{2}},   
\end{equation}
where $a_{j,t} = 0$ indicates that the signal $j$ has been successfully decoded and canceled under the condition that its SINR is larger than a threshold $\Gamma$, i.e.,  $\mathbf {SINR_{j,t}} \geq \Gamma$; and $a_{j,t} = 1$ otherwise. It is worth noting that for NOMA transmissions, the signal from vehicle $m$ can only be decoded after signals with greater received strengths have been successfully decoded and canceled., i.e., $a_{j,t}=0, \forall j=\{1,2,...m-1\}$~\cite{chen2022deep}. 
Therefore, provided that the preceding signals have been successfully decoded, \eqref{eq:SINR_before} can be rewritten as
\begin{align}
\mathbf {SINR}_{n,t}=\frac{p^r_{n,t}}{\sum_{i=n+1}^{M} p^r_{i,t}+\sigma_t^{2}}.
\end{align}

The access indicator of the $n$-th vehicle at the time slot $t$ under the SINR constraint can be written as: 
\begin{equation}
   I_{n,t}=
   \begin{cases}
   1, & \text{$\mathbf {SINR}_{n,t}\geq \Gamma$},\\
   0, & \text{$\mathbf {SINR}_{n,t}< \Gamma$.}
   \end{cases}
\end{equation}
The number of vehicles access the network at time slot $t$ is given by
$
M_t = \sum_{n = 1}^N I_{n,t}.
$



To enhance the join ratio of the FL system, we aim to maximize $M_t$, which represents the number of vehicles that can successfully upload their local model information to the BS where a global model is implemented at time slot $t$.

\label{sec:system model}

\section{Vehicle Selection and Transmission Power Allocation}  

In this paper, the strategy to maximize $M_t$ is outlined with an emphasis on the proper selection of vehicles and the careful design of the transmission power for each selected vehicle. This approach ensures that the $\mathbf{SINR}$ requirements are satisfied, facilitating the successful transmission of local model information to the BS.

Denote the maximum transmit power of the $n$-th vehicle as $p^t_{n,max}$, and the channel state information (CSI) between the $n$-th vehicle and the base station (BS) as 
$h_{n,t}$. Based on the CSI, the central controller determines  the transmit power of each vehicle, $p^t_{n,t}$, thereby ensuring that the received power $p^r_{n,t}$ remains within a certain bound.

\begin{lemma}\label{lemma:decode_constraint}
Provided that the signal from the $(n+1)$-th vehicle is decodable, the signal from the $n$-th vehicle will be decodable if the following conditions are met: $\frac{p^r_{n,t} }{p^r_{n+1,t}}\geq (1+\Gamma), \forall n<M_t-1$; and $p^r_{M_t-1,t}\geq (p^r_{M_t,t}+\sigma_t^2)\Gamma$.
\begin{proof}     

If the signal of the $(n+1)$-th vehicle can be decoded, i.e., $\mathbf{SINR}_{n+1,t}\geq \Gamma, \forall k<M_t-1$, the interference and noise of  the $(n+1)$-th signal, denoted as $\mathcal{I}_{n+1,t}$, is upper bounded by: 
\begin{align}\label{eq:Nk+1}
    \mathcal{I}_{n+1,t}= \sum_{i=n+2}^{M_t}p^r_{i,t} +\sigma_t^2 \leq \frac{p^r_{n+1,t}}{\Gamma}, \forall n < M_t.
\end{align}
If the condition $\frac{p^r_{n,t} }{p^r_{n+1,t}}\geq (1+\Gamma)$ satisfies, the interference of the signal for the $n$-th vehicle,   $\mathcal{I}_{n,t}$, is upper bounded by: 
\begin{align}\label{eq:d1}
     \mathcal{I}_{n,t} &= \sum_{i=n+1}^{M_t} p^r_{i,t}+\sigma_t^2
     \leq p^r_{n+1,t}+\frac{p^r_{n+1,t}}{\Gamma}\nonumber\\&
     =p^r_{n+1,t}\left(1+\frac{1}{\Gamma}\right)\leq \frac{p^r_{n,t}}{\Gamma}, \forall n< M_t-1.
\end{align}
which ensures $\mathbf{SINR}_{n,t}\geq \Gamma, \forall n<M_t-1$.

When it comes to the last two signals, to ensure $\mathbf{SINR}_{M_t-1,t} \geq \Gamma$, we have $p^r_{M_t-1,t}\geq (p^r_{M_t,t}+\sigma_t^2)\Gamma$. This completes the proof. 

\end{proof}
\end{lemma}

\begin{algorithm}
	\renewcommand{\algorithmicrequire}{\textbf{Input:}}
	\renewcommand{\algorithmicensure}{\textbf{Output:}}
	\caption{Vehicle Selection and Power Control algorithm}
	\label{alg:alg1}

 \textbf{Input:}  $ p^t_{n,t}, h_{n,t}, N, \sigma_t^2,\Gamma$. \\ 
    Initialize sets $p^t_{n,t}=p^r_{n,t}=\epsilon=0$, $i=j=l=1$ and $U_t=\Phi = \Lambda = \emptyset$. 
	\begin{algorithmic} [1]

      \STATE Decide the maximum achievable received signal strength of each vehicle based on their available energy, $p^{r}_{n,max}=p^t_{n,max}h_{n,t} , \forall n.$

		\STATE  Sort $p^{r}_{n,max}$ in the descending order $p^{r}_{1,max}>...>p^{r}_{M_t,max}$.

\IF{$p^{r}_{1,max}\geq \sigma_t^2 \Gamma$}
\STATE 
 Set $p^{r}_{1,t}=p^{r}_{n,max}$ and add vehicle $1$ into $U_t$;
\ENDIF
 \FOR{$n=2:N$}

	\IF{$\frac{p^r_{n-1,t}}{\Gamma+1} > {\sigma_t^2\Gamma}$}

\STATE Add vehicle $n$ into $U_t$;

    \IF{$\frac{p^r_{n-1,t}}{\Gamma+1}< p^r_{n,max}$}
    \STATE Set $p^r_{n,t}=\frac{p^r_{n-1,t}}{\Gamma+1}$;
    \STATE Calculate $\phi_{n,t}=p^r_{n,max}-\frac{p^r_{n-1,t}}{\Gamma+1}$, and set $\lambda_i=n$; i++;
   
    \ELSE

\STATE Set $p^r_{n,t}=p^r_{n,max}$, $\phi_{n,t}=0$ and update the power difference $\epsilon=\frac{p^r_{n-1,t}}{\Gamma+1}-p^r_{n,max}$;

\WHILE{($\sum \lambda \neq 0$ \& $\epsilon\neq 0$)}

\STATE  $j=\lambda_l$;
	
        \IF{$\phi_{n,t} > \epsilon(1+\frac{1}{\Gamma})$}
 \STATE	$p^r_{j,t}=p^r_{j,t}+\epsilon(1+\frac{1}{\Gamma})$; $\phi_{j,t}=\phi_{j,t}-\epsilon(1+\frac{1}{\Gamma})$; $\epsilon=0$
  \ELSE
 \STATE $p^r_{j,t}=p^r_{j,t}+\phi_{j,t}$; $\epsilon=\epsilon-\frac{\phi_{j,t}\Gamma}{1+\Gamma}$; $\phi_{j,t}=0$; $\lambda_l=0$; $l$++;
 	\ENDIF
 	
 \ENDWHILE
 	\ENDIF
 \ENDIF

       \ENDFOR
    
\STATE  Calculate $p^t_{n,t}$ based on channel inversion, $p^t_{n,t} = \frac{p^r_{n,t}}{h_{n,t}}, \forall n$;
              
		\ENSURE   $ p^t_{n,t}$, $U_t$
	  \end{algorithmic}  
\end{algorithm}

\begin{lemma}
\label{lemma:power_level}
Given the largest received power at the BS is $p^r_{1,t} = p^t_{1,max}h_{1,t}$, the maximum joining ratio that can participate in NOMA-enabled FVN is upper bounded by 
\begin{align}\label{powerlevel_max}
    \bar{J} \leq \frac{\lfloor \log{(\frac{p^r_{1,t}}{\sigma_t^2\Gamma})} -  \log{(1+\Gamma)}+1\rfloor}{N},
\end{align}
in which $\lfloor x \rfloor$ denotes the largest integer smaller than or equal to $x$.
\end{lemma}
\begin{proof}

According to $\textbf{Lemma}~\ref{lemma:decode_constraint}$, we can observe that whether the $n$-th signal can be decoded depends mainly on the power of the following signals. Thus, if we follow the design that $p^r_{n,t} = p^r_{n+1,t}(1+\Gamma)$, and $p^r_{M_t-1,t}= (p^r_{M_t,t}+\sigma_t^2)\Gamma$, to ensure $\mathbf{SINR}_{M_t,t} = p^r_{M_t,t}/\sigma_t^2 \geq \Gamma$, we have that $p^r_{1,t} = p^r_{M_t,t}(1+\Gamma)^{\bar{M_t}-1}\geq\sigma_t^2 \Gamma(1+\Gamma)^{\bar{M_t}-1}$,
 Thus, the upper bound of the maximum number of vehicles that can be successfully transmitted simultaneously is given by: 
\begin{align}\label{eq:upperbound_numberofdevices}
\bar{M_t}&\leq  \log{(\frac{p^r_{1,t}}{\sigma_t^2\Gamma})} -  \log{(1+\Gamma)}+1 .
\end{align}
Thus, the maximum joining ratio can be written as equation \eqref{eq:upperbound_numberofdevices}.
\end{proof}

\vspace{-6pt}
\begin{lemma}
   Consider a scenario with $M_t$ vehicles, where the received signal strengths are defined as follows $p^r_{1,t}$, $p^r_{2,t} \leq \frac{p^r_{1,t}}{1+\Gamma}$, ..., $p^r_{M_t-1,t} \leq (\frac{p^r_{{M_t-2},t}}{1+\Gamma})$, $p^r_{M_t,t} \leq (\frac{p^r_{{M_t-1},t}}{\Gamma}-\sigma_t^2)$.  If the signal from vehicle $n$ is reduced by $\epsilon$, then all subsequent vehicles must also decrease their signal strengths according to Lemma~\ref{lemma:decode_constraint} to ensure that the signal of vehicle $n$ can be decoded. A subset of vehicle $X, \forall x<n,$ can increase their received signal strengths, if applicable, by up to $\epsilon (1+\frac{1}{\Gamma})$ such that signals of all vehicles can still be decoded. 
\begin{proof}     
According to~\eqref{eq:d1}, the tolerable interference of vehicles before $n$ is upper bounded by $\mathcal{I}_{j,t}=\sum_{i=j+1}^{m} p^r_{i,t}+\sigma_t^2\leq p^r_{j+1,t}(1+\frac{1}{\Gamma})$ to ensure $\mathbf{SINR}_{j,t} \geq \Gamma, \forall j< n.$ If the signal strength of vehicle $n$ is reduced by $\epsilon$, the interference of vehicles before $n$, i.e., $\mathcal{I}_{j,t}, \forall j< m$, decrease by $\epsilon(1+\frac{1}{\Gamma})$ accordingly. Therefore, it is both desirable and feasible to enhance the signal powers of vehicles preceding vehicle $n$ by a cumulative amount of 
 $\epsilon(1+\frac{1}{\Gamma})$.  This increase should still comply with the SINR requirements of all vehicles involved.

Denote the adjusted signal strengths as ${\hat{p^r}_{1,t},\dots,\hat{p^r}_{n,t},\dots,\hat{p^r}_{M_t,t}}$, where $\hat{p^r}_{j,t}\geq{p^r_{j,t}}, \forall j<n$, $\hat{p^r}_{n,t} = p^r_{n,t}-\epsilon$ and $\hat{p^r}_{i,t}\geq{p^r_{i,t}}, \forall i>n$. For any vehicle $j$ before $n$, we have $\sum_{z =j+1}^{n-1}\hat{p^r}_{z,t}\leq \sum_{z =j+1}^{n-1}q_{z,t} +\epsilon(1+\frac{1}{\Gamma})$, and $\sum_{i=n}^{m_t}\hat{p^r}_{i,t} \leq (p^r_{n,t}-\epsilon)(1+\frac{1}{\Gamma})$. Thus, $\sum_{j=i+1}^{M_t}\hat{p^r}_{j,t}+\sigma_t^2\leq \sum_{j =i+1}^{n}p^r_{j,t}+\sum_{j=n+1}^{M_t}p^r_{j,t}+\sigma_t^2$.
Notice that $\hat{p^r}_{j,t}\geq{p^r_{j,t}}, \forall j<n$, the SINR of vehicle $j$: 
 \begin{align}
     \mathbf{SINR}_{j,t} &= \frac{\hat{p^r}_{j,t}}{\sum_{i=j+1}^{M_t}\hat{p^r}_{i,t}+\sigma_t^2}\nonumber\\
     & = \frac{\hat{p^r}_{j,t}}{\sum_{i =j+1}^{n-1}\hat{p^r}_{i,t}+\sum_{z=m}^{M_t}\hat{q}_{z,t}+\sigma_t^2}\nonumber\\&\geq  \frac{p^r_{j,t}}{\sum_{i=j+1}^{n-1}p^r_{i,t}+\sum_{z=m}^{M_t}p^r_{z,t}+\sigma_t^2}\geq \Gamma.
 \end{align}
The SINR requirements of vehicles equal to and after $n$ are satisfied according to Lemma~\ref{lemma:decode_constraint}.  
 \end{proof}
\end{lemma}

To achieve maximum $M_t$, we propose a heuristic algorithm in \textbf{Algorithm~1}.  
We first determine the maximum achievable signal strength for each vehicle, denoted as, $p^{t}_{n,max}$, using their maximum available transmission power. The vehicles are then organized in descending order based on  $p^r_{n,max} = p^t_{n,max}h_{n,t}$. If the first vehicle in this sequence meets the SINR threshold, i.e., $ \mathbf{SINR}_{1,t}\geq \Gamma$, then the maximum power should be applied for transmission, $p^r_{1,t}=p^t_{1,t}h{1,t}$, the vehicle should then be added to the transmission set $U_t$. To guarantee the successful decoding of vehicle 1, the transmission powers of the following vehicles need to be set based on certain requirements $p^r_{n,t}/p_{n+1,t}\geq (1+\Gamma)$ according to $\textbf{Lemma 1}$, subject to the available energy $p^{r}_{n,max}.$  If $p^{r}_{n,max} \geq \frac{p_{n-1,t}}{\Gamma+1}$, vehicle $n$ has sufficient power but should use $p_{n,t}^r = \frac{p_{n-1,t}^r}{\Gamma+1}$ to ensure previous vehicles can be decoded successfully, in which case vehicle $n$ has extra power $\phi_{n,t}$. Otherwise, vehicle $n$ uses its maximum power for transmission, $p^r_{n,t} =p^{r}_{n,max}$, and the power difference is $\epsilon= (\frac{p_{n-1,t}}{\Gamma+1} - p^{r}_{n,max})$. Based on $\textbf{Lemma 3}$, if vehicle $n$'s signal is reduced by $\epsilon$, any vehicles before $n$ can increase their signal strength by up to $\epsilon(1+\frac{1}{\Gamma})$ such that all vehicles can be decoded. Thus, we check all vehicles before $n$ that have extra powers, i.e., $\forall j < n$, and $\phi_{j,t}>0$.

A vehicle $j$ can increase its signal power given that the increased signal strength does not exceed the available extra power $\phi_{j,t}$, and the total increased signal strength does not exceed $\epsilon(1+\frac{1}{\Gamma})$. The round of power increase for vehicles before $k$ completes when $\epsilon(1+\frac{1}{\Gamma})$ is added to vehicles before $n$ and $\epsilon$ is reset to $0$; or all extra power of vehicles are used up, $\phi_{n,t}=0, \forall j<n$. We continue to check all vehicles and add those with satisfactory SINR into $U_t$, and adjust their transmission powers based on the received signal $p^t_{n,t} = \frac{p^r_{n,t}}{h_{n,t}} \forall n.$

\section{Federated Learning} \label{sec:algorithms}
In this section, we develop a NOMA-enabled FL (NFL) algorithm that extends the traditional Federated Averaging (FedAvg) algorithm, as shown in \textbf{Algorithm~2}, by integrating the user selection and power control for NOMA transmissions described in Sec. III.

\textbf{Local model initialize.} We denote the training iteration as $t \in \{0,1,2, ..., T-1\}$. The central server first sends the initialized model to each vehicle to initialize the local models of all FL vehicles in the current iteration using the downlink communication channel, i.e., 
\begin{equation}
    \omega_{n}^{(t,0)} \leftarrow \omega^{(t)} \enskip \forall n \in U_t,
\end{equation}
where $n$ is the index of $n$-th selected vehicle, and $\omega_{n}^{(t,0)}$ represents, $n$-th vehicle's initial model at the $t$-th communication iteration. Additionally, $\omega^{(t)}$ is the global model.

\noindent
\textbf{Local model update.} For each of the selected FL vehicles in $U_t$, they apply stochastic gradient descent (SGD) to minimize their loss function $\textit{f}_{n}: \mathcal{X}\times \mathcal{Y} \to \mathbb{R}_{+}$ in $\tau$ iteration of local training. The $n$-th FL vehicle's local model update can be written as:
\begin{equation}
    \omega_{n}^{(t,\tau+1)} \leftarrow \omega_{n}^{(t,\tau)} - \frac{1}{k_n}\eta\nabla f_{n}(\omega_{n}^{(t, \tau)}, \xi_{n}^{t}),
\end{equation}
where $\xi_{n}^{t}$ is a sample uniformly chosen from the local data $D_n$ at the $t$-th iteration and $\eta$ is the learning of local training.

\noindent
\textbf{Central BS model aggregation} After local training, the selected FL vehicles send their models to the server at BS using the uplink wireless communication channel. The server aggregates those local models for the next communication iteration as follows:
\begin{equation}
    \omega^{(t+1)} \leftarrow \sum_{n \in U_t} \alpha_n\omega_{n}^{(t,\tau+1)},
\end{equation}
where $\alpha_n$ is the weight of the $n$-th vehicle such that $\alpha_n > 0$ and $\sum_{n \in U_t} \alpha_n = 1$.

\begin{algorithm}
	\renewcommand{\algorithmicrequire}{\textbf{Input:}}
	\renewcommand{\algorithmicensure}{\textbf{Output:}}
	\caption{NFL Algorithm}
	\label{alg:alg2}

 \textbf{Input:} Learning rate $\eta$, initialize global model $\omega^{(0)}$, selected FL vehicles set $U_t$, number of communication rounds $T$, local dataset $\xi_{n}^{(t)}$ for $t \in \{0,1,...,T-1 \}$.
\begin{algorithmic} [1]
\FOR{$t \in \{0,1,...,T-1 \}$}
\STATE Obtain the selected FL vehicles set $U_t$;\\
    \FOR{Vehicle $n \in U_t$}
       \STATE Let $ \omega_{n}^{(t,0)} \leftarrow \omega^{(t)}$;
       \FOR{$\tau \in \{0,1,2,...,\tau^{*}-1\}$}
        \STATE   Local model update: \\  
        $\omega_{n}^{(t,\tau+1)} \leftarrow \omega_{n}^{(t,\tau)} - \frac{1}{k_n}\eta\nabla f_{n}(\omega_{n}^{(t, \tau)}, \xi_{n}^{t})$;
        \ENDFOR
    \ENDFOR
    \FOR{BS received all local models from $n \in U_t$}
      \STATE Update global model: $\omega^{(t+1)} \leftarrow \sum_{n \in U_t} \alpha_n\omega_{n}^{(t,\tau^{*})}$;
    \ENDFOR
    
\ENDFOR
\ENSURE Optimal global model: $\omega^{*}$
\end{algorithmic}
\end{algorithm}
\noindent

\section{Performance Evaluation} \label{sec:numerical}
\begin{figure}[t]
    \includegraphics[width=0.37\textwidth]{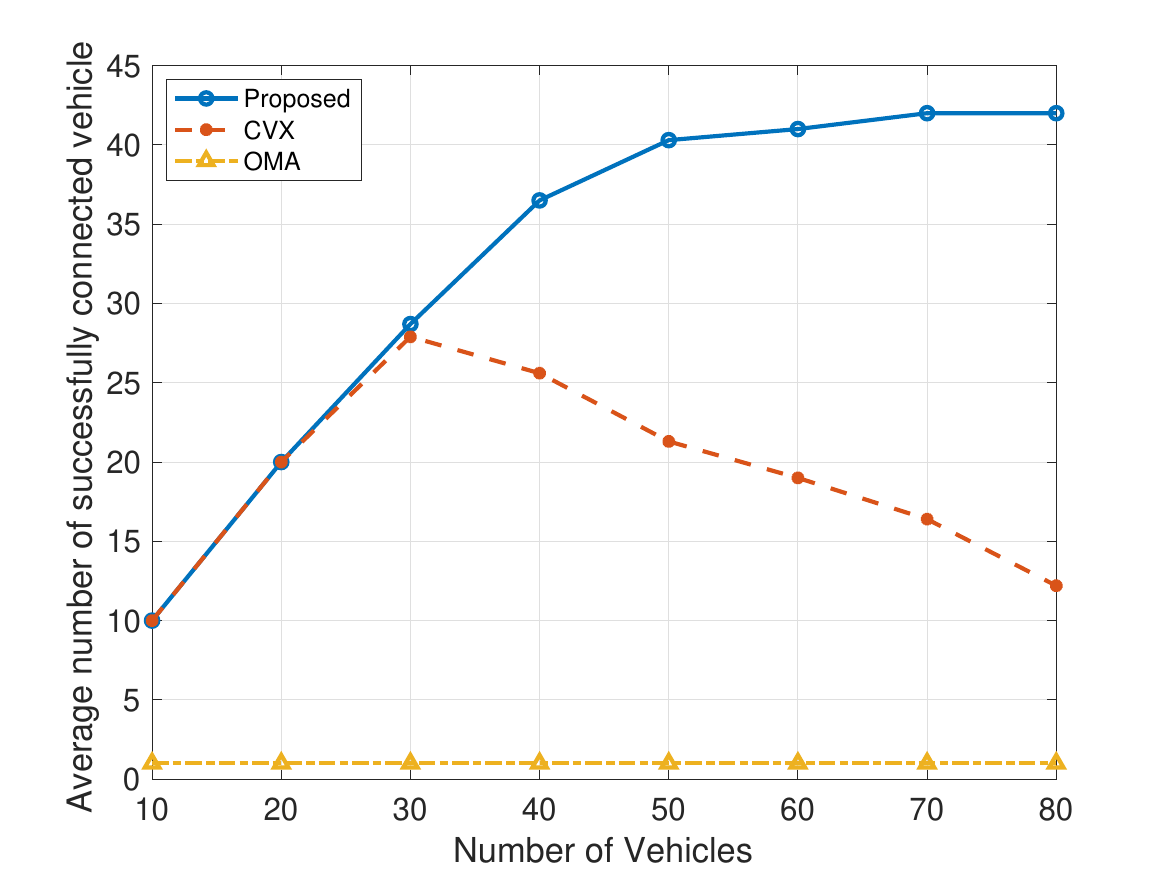}
    \centering
    \caption{Average number of successfully connected vehicles versus the total number of vehicles in the network.}
    \label{fig:noma}
\end{figure}

\begin{figure}
    \includegraphics[width=0.37\textwidth]{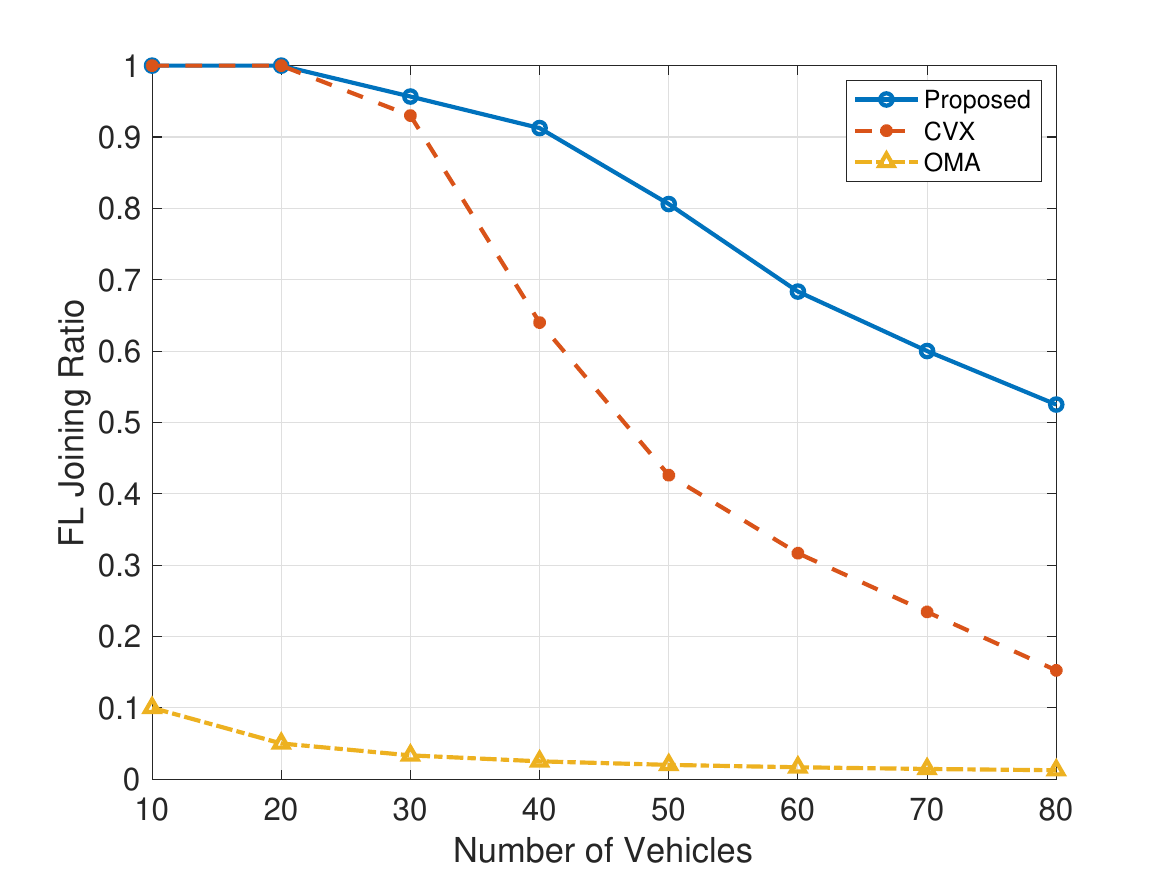}
    \centering
    \caption{Average number of successfully connected vehicles versus the total number of FL joining ratio.}
    \label{fig:jr}
\end{figure}

\begin{figure*}[hthp]
  \centering
  \hfill
 \begin{minipage}{0.24\textwidth} 
    \centering
\subfloat[Test Accuracy of i.i.d. data]{\includegraphics[width=\textwidth]{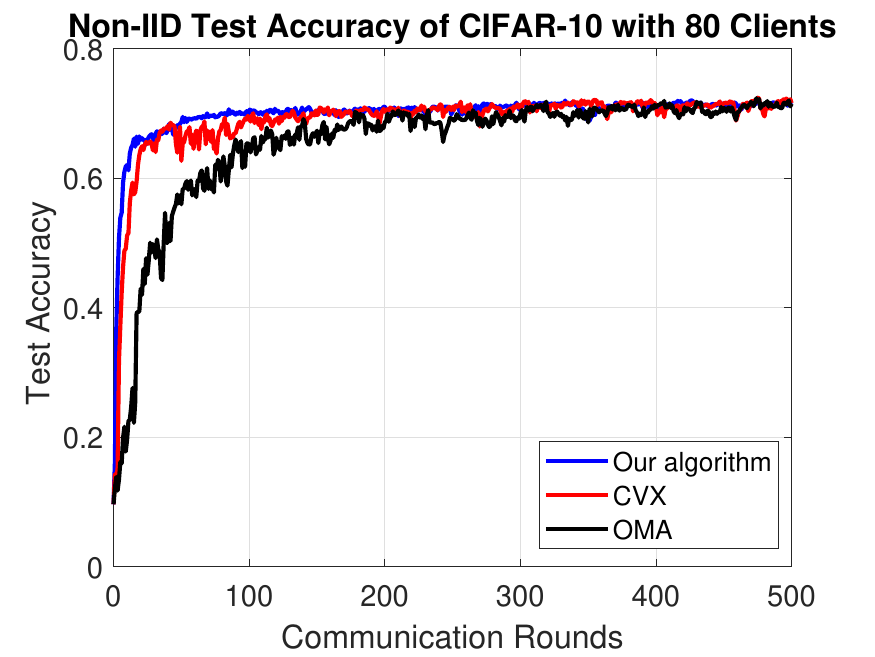}}  
    \label{fig:iidAcc}
  \end{minipage}
  \hfill
 \begin{minipage}{0.24\textwidth}
  	\centering
    \subfloat[Test Accuracy of non-i.i.d. data.]{\includegraphics[width=\linewidth]{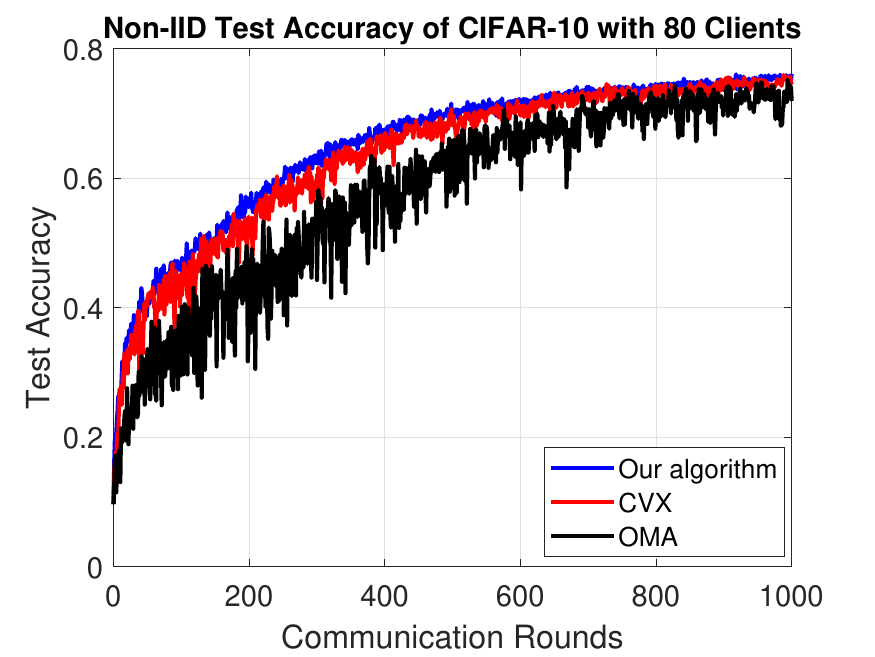}}  
\label{fig:geo_th}
  \end{minipage}
  \hfill
  \begin{minipage}{0.24\textwidth}
    \centering
    \subfloat[Train Loss of i.i.d. data.]{\includegraphics[width=\linewidth]{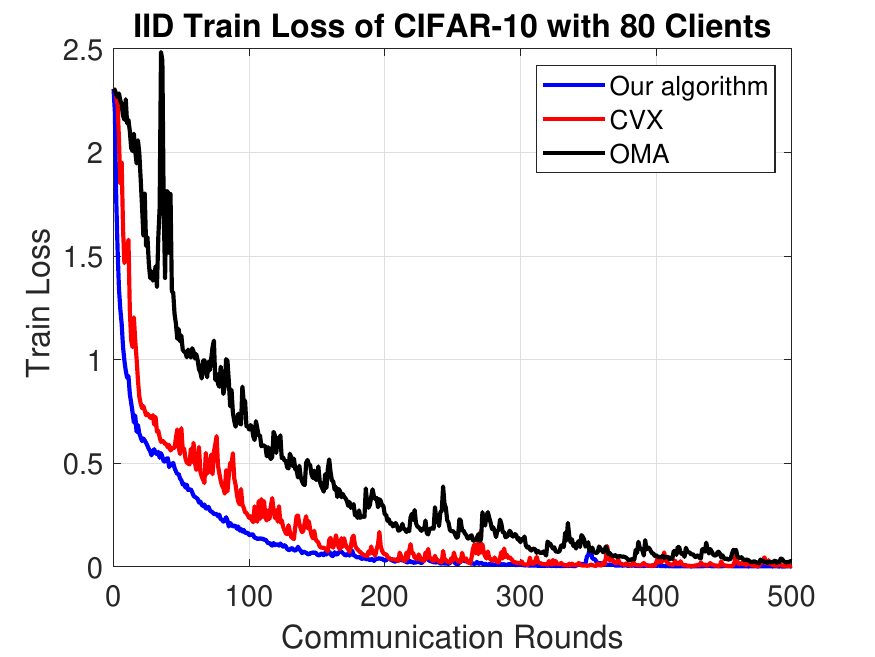}} 
\label{fig:iidLoss}
  \end{minipage}
  \hfill
  \begin{minipage}{0.24\textwidth}
    \centering
    \subfloat[Train Loss of non-i.i.d. data.]{\includegraphics[width=\linewidth]{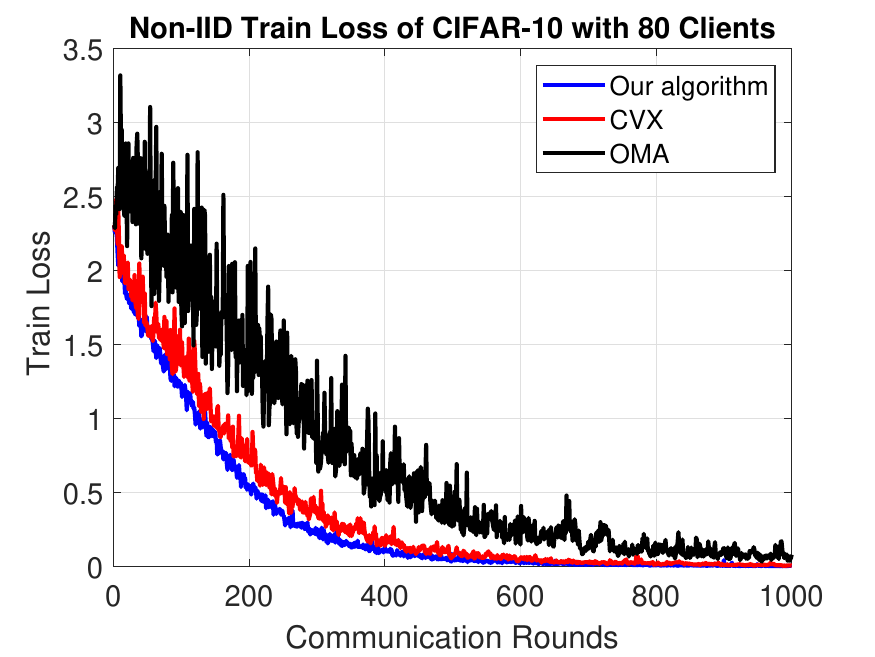}} 
\label{fig:mm_th}
  \end{minipage}
  \caption{Federated Learning results with CIFAR-10 data.}
  \label{fig:NoniidLoss}
  \hfill
\end{figure*}

The vehicle trajectory data utilized in this research was collected through a volunteer-based data collection initiative. This dataset provides a comprehensive list of activities for each individual, recorded at three-second intervals. Specifically, the study extracted and employed the data corresponding to the specific moment with the highest volume of vehicles captured within a 6-mile radius of Elmhurst, IL. To evaluate the impact of varying the number of FL clients, which we refer to as the joining ratio, we utilize the widely recognized CIFAR-10 image classification dataset for training.  We consider a total of 80 vehicles in the network and for each vehicle in the coverage area of the BS, we partition the local data into training and testing sets, i.e., allocating 75\% of the data for training and 25\% for testing.  Additionally, we consider both i.i.d. and non-i.i.d. data sets to show the impacts of the joining ratio on the overall learning performance of the FNV. Moreover, the non-i.i.d. datasets of vehicles follow the Dirichlet distribution with $\alpha_{d} = 0.4$. We apply the FedAvg algorithm for FL with learning rate $\eta = 0.05$. The training process consists of 500 communication rounds for the i.i.d. scenario and 1,000 rounds for the non-i.i.d. scenario to ensure convergence by using a 3-layer Resnet neural network. Without loss of generality, Rayleigh fading is considered for the communication links between vehicles and the base station.  

Our analysis starts by comparing our vehicle selection and power control algorithm with the results obtained from orthogonal multiple access (OMA) methods and a referenced power allocation algorithm which is done by using CVX~\cite{zeng2020sum}. Fig.~\ref{fig:noma} displays the number of successfully connected vehicles versus the total number of vehicles in the network. We can observe that, compared to other solutions, our algorithm can support the largest number of vehicles to upload the local model information. 
Specifically, due to resource limitations, OMA can connect only one vehicle at a time. In contrast, the CVX algorithm, as discussed in~\cite{zeng2020sum}, supports multiple vehicles simultaneously. However, this approach also results in a higher outage probability when the number of vehicles is high. In~\cite{zeng2020sum}, 
it is required that the SINR of all vehicles involved in NOMA transmissions exceed a threshold, $\Gamma$, which can lead to errors in CVX when no solution can be found, and leads to the performance degradation when the number of vehicles are increased in the network. On the other hand, our algorithm not only controls transmit power but also assists in vehicle selection, thereby ensuring successful transmissions consistently. Additionally, Fig. \ref{fig:jr} also shows that our proposed NFL algorithm can achieve a higher joining ratio compared to CVX and OMA.

With the output from the FL vehicle selection with power control, we further apply the NFL algorithm to evaluate the training performance. We present the comparison of distributed learning algorithms over i.i.d.  and non-i.i.d. data distributions using the CIFAR-10 dataset with 80 clients. Our proposed NFL algorithm outperforms the baseline algorithms CVX and OMA, achieving higher test accuracy with less variability across communication rounds. Specifically, in the i.i.d. scenario, our algorithm converges faster to a test accuracy near 0.7, indicating a robust and effective learning process. When dealing with the non-i.i.d. scenario, which better reflects real-world data distributions, our algorithm shows its resilience and adaptability, shown in Fig. \ref{fig:NoniidLoss}, reaching similar accuracy levels as in the i.i.d. scenario, though requiring more communication rounds due to the inherent complexity of non-uniform data. 
Additionally, the train loss for both i.i.d. and non-i.i.d. scenarios show superior efficiency and stability compared with the other two baselines. 
The advantage of our approach is further highlighted in the non-i.i.d. environment, where despite the increased complexity and the naturally extended convergence process, our algorithm markedly outperforms the baselines, achieving a lower and more stable training loss. Our algorithm's enhanced stability and performance, especially evident in the non-i.i.d. scenario, underscore the potential for improved FL performance.






\section{Conclusion} \label{sec:conclusion}
In this paper, we have proposed a NOMA-enabled FVN frame. A joint vehicle selection and power allocation algorithm has been developed to improve the joining ratio of the FVN. Extensive simulations have shown that our algorithm not only significantly increases the joining ratio but also substantially enhances the overall performance of the FVN. In our future work, we plan to explore adaptive federated learning for services beyond CIFAR-10 such as real-time traffic monitoring and management. 

\bibliographystyle{IEEEtran}
\bibliography{IEEEfull,reference}
\end{document}